\documentclass{iopart}

\usepackage{iopams}

\expandafter\let\csname equation*\endcsname=\relax 
\expandafter\let\csname endequation*\endcsname=\relax 
\usepackage{amsmath,amssymb,amsthm}

\usepackage{stackrel}
\usepackage[utf8]{inputenc}
\usepackage{tikz}
\usepackage{mathrsfs}
\usepackage{epsf}
\usepackage{bookmark,hyperref}
\hypersetup{
     colorlinks=true,
     linkcolor=blue,
     filecolor=blue,
     citecolor = blue,
     urlcolor=cyan,
     }
\usepackage{pstricks}

\usepackage[perpage]{footmisc}

\newcommand{\N}{{\mathbb N}}

\newcommand{\R}{{\mathbb R}}

\newcommand{\gi}{\mathrm{ g}}

\newtheorem{teo}{Theorem}[section]
\newtheorem{lema}[teo]{Lemma}
\newtheorem{cor}[teo]{Corollary}
\newtheorem{prop}[teo]{Proposition}

\newtheorem{rmk}{Remark}[section]

\begin{document}

\title[A Holonomic Rattleback]{A Holonomic Rattleback}

\author{J. M. Burgos}

\address{Departamento de Matem\'aticas, CINVESTAV--\,CONACYT, Av. Instituto Polit\'ecnico Nacional 2508, Col. San Pedro Zacatenco, 07360 Ciudad de M\'exico, M\'exico.}
\eads{\mailto{burgos@math.cinvestav.mx}}

\begin{abstract}
The rattleback or celt top is a rigid body with the peculiar behaviour of having a spontaneous spin reversion during its motion and this effect is usually attributed purely to its nonholonomic nature. Actually, the rattleback is the paradigmatic example in nonholonomic mechanics. We give an example of a spinning rigid body having spontaneous spin reversion during its motion in the context of holonomic mechanics.
\end{abstract}


\maketitle


\section{Introduction}

The \textit{rattleback}, also known as the \textit{celt top}, is a nonholonomic mechanical system with the curious effect of spontaneous reversion of its spin. The usual explanation is to attribute the effect to the nonholonomic nature of the system arguing that the nonholonomic constraint causes oscillations that eventually consume the spin energy in such a way that after it vanishes, the oscillations decay and the top spins in the other direction. Actually, the rattleback is one of the paradigmatic examples in nonholonomic mechanics.

Since 1895 in \cite{16}, this unintuitive effect has been studied by many authors with radically different approaches and their common feature is to exploit the nonholonomic nature of the system, see \cite{9}, \cite{3}, \cite{7}, \cite{2}, \cite{15}, \cite{8}, \cite{12}, \cite{4} and
the references therein. A different approach was taken in \cite{Integrable_rattleback} where the dynamics of a completely integrable version of the rattleback introduced by the authors is described.

Here, we reproduce the spin reversion effect in the holonomic context by considering an asymmetric rigid body as in the usual case but substituting the nonholonomic contact point on a flat surface by holonomic constraints. We emphasize the non integrability of our holonomic model in contrast with the model  in \cite{Integrable_rattleback}. We also emphasize that the spontaneous spin reversion in our model has nothing to do with the Dzhanibekov effect also known as the intermediate axis Theorem due to the instability of the rotation along the axis with intermediate inertia moment \cite{Tennis_racket}, (\cite{LL}, Section 37).

Now we describe our model. Consider a rigid body with diagonal moment of inertia tensor with inertia moments $i_x$, $i_y$ and $i_z$ along the respective axes $x$, $y$, and $z$. The first holonomic constraint is the fixed center mass and the second holonomic constraint is the configuration space given by rotations of the form
$$R(\gamma,\,\alpha)=X(\gamma)\,Z(\alpha)$$
where $X(\gamma)$ and $Z(\alpha)$ denote the rotation along the $x$ and $z$ axes with angles $\gamma$ and $\alpha$ respectively. The angular coordinate $\gamma$ is subject to a simple harmonic oscillator potential with Hooke's constant $k>0$
\begin{equation}\label{Hooke_constant}
U_k(\gamma)\,=\,k\,\frac{\gamma^2}{2}.
\end{equation}

The point with null initial conditions is an equilibrium point. Moreover, the uniform rotation along the $z$ axis is a solution of the system corresponding with any Hooke's constant $k>0$
$$\gamma_k(t)=0,\qquad \alpha_k(t)=\dot{\alpha}_0\,t.$$
Another more interesting although also trivial solution of the system is in the case where there is an identity of the inertia moments $i_x=i_y$ and the motion decouples in a uniform rotation along the $z$ axis together with a simple harmonic oscillator motion in the angular $\gamma$ coordinate. These solutions are reminiscent of the holonomic nature of the system.

However, if the rigid body is asymmetric with $i_y<i_x$ and a relation between the initial velocity and the inertia moments holds, then its motion has  spontaneous spin reversion. This clearly confronts the paradigmatic classical view about the role of the nonholonomic nature on the spontaneous spin reversion of the rattleback.

\begin{teo}[Rattleback effect]\label{main}
Consider a Hooke's constant $k>0$ and the motion starting at the point with null angular coordinates and non null initial angular velocities $\dot{\gamma}_0$ and $\dot{\alpha}_0$. If
$$i_y\,<\,i_x\qquad\mbox{and}\qquad \frac{\dot{\alpha}_0^2}{\dot{\gamma}_0^2}\,<\,\frac{i_x}{i_z}\,\left(\left(\frac{i_x}{i_y}\right)^{1/2}\,-\,1\right),$$
then for every natural $n$ there is $k_0>0$ such that for every $k\geq k_0$ the motion has at least $n$ spontaneous spin reversions.
\end{teo}

This phenomenon is due only to the asymmetry of the body and the existence of a constant $\theta$ such that the following family of smooth functions constructed from the solutions of the system weakly star converges to it as $k$ goes to infinity
$$\dot{\gamma}_k^2\,\left(i_x\,\cos^2(\alpha_k)\,+\,i_y\,\sin^2(\alpha_k)\right)^{3/2}\,\rightharpoonup\, \theta.$$
That is to say, on every open bounded set the average on the right hand side converges to the average on the left. This will be proved in subsection \ref{Proof_Strong}, Lemma \ref{Lema_weak_clave}, see also Remark \ref{Remark2}. In (\cite{Bornemann}, Section 2.6) using a result in (\cite{Arnold_Kozlov}, Section 5.4) it is shown the relation of this weak convergence with the classical notion of \textit{adiabatic invariance}. We emphasize the purely holonomic nature of the effect.

The right hand side of the relation in Theorem \ref{main} is a critical value of a phase transition in the dynamics of the rattleback. Indeed, we prove the following

\begin{teo}[No rattleback effect]\label{main2}
Consider a Hooke's constant $k>0$ and the motion $(\gamma_k(t),\,\alpha_k(t))$ starting at the point with null angular coordinates and non null initial angular velocities $\dot{\gamma}_0$ and $\dot{\alpha}_0$. If $i_y\,\geq\,i_x$ or
$$i_y\,<\,i_x\qquad\mbox{and}\qquad \frac{\dot{\alpha}_0^2}{\dot{\gamma}_0^2}\,>\,\frac{i_x}{i_z}\,\left(\left(\frac{i_x}{i_y}\right)^{1/2}\,-\,1\right),$$
then for every $A>0$ there are $T>0$ and $k_0>0$ such that for every $k\geq k_0$ we have $\vert\,\alpha_k(T)\,\vert> A$ and $\alpha_k$ is strictly monotonic on the interval $[0,T]$.
\end{teo}

The strategy of the proof is to consider the limit as the Hooke's constant goes to infinity as a \textit{strong constraining problem}. This problem concerns the question of whether the theory of \textit{ideal constraints}, that is to say those verifying the \textit{D'Alembert-Lagrange principle}, is recovered as the limit of stiff potentials that force the motion to the constraint, that is to say whether the motion equations resulting from the limit of stiff potentials coincide with those resulting from the least action principle subject to the respective ideal constraints. We will refer to these as \textit{real constraints}. See the sections 3.6 and 3.10 in Gallavotti's book \cite{Gallavotti} for a nice discussion on ideal and real constraints as well as for historical references.

The answer to the previous question is in general no and the responsible is the high frequency behaviour at the limit. The problem of obtaining the right description of the motion for a real constraint at this high frequency limit is part of the strong constraining problem. 

The first work on strong constraining is due to Rubin and Ungar in \cite{RubinUngar}. Later, claiming that his ``research started with an attempt to clarify some provoking remarks in (\cite{Arnold}, sections 17.A and 21.A)", Takens in \cite{Takens} studies the strong constraining of a family of quadratic potentials verifying a non resonant condition. He also shows that resonance leads to funnels in the limit motion also known as Takens chaos. In \cite{BS}, Bornemann and Sch\"utte used weak convergence methods to study the problem for Hamiltonian systems in codimension one. In \cite{Bornemann}, Bornemann generalized the previous result to arbitrary codimension and generalized the non resonant condition as well. All of the previous considered potentials are non degenerate at the constraint. The problem of strong degenerate constraining in codimension one is part of a current research by the author, see \cite{Burgos}.

As a concrete physical example of a high frequency limit, consider the Kapitza's inverted pendulum (\cite{Arnold}, section 25.E)
$$\ddot{\theta}_\varepsilon = \left( a+\frac{b}{\varepsilon}\cos\left(\frac{t}{\varepsilon}\right)\right)\sin(\theta_\varepsilon),\ \ \ \theta_\varepsilon(0)=\alpha,\ \dot{\theta}_\varepsilon(0)= \beta,\ t\geq 0,\ \varepsilon>0.$$
As $\varepsilon$ goes to zero, the solutions $\theta_\varepsilon$ present a high frequency behaviour and converge locally uniformly to the solution of the problem
$$\ddot{\theta}= a\sin(\theta)- \frac{b^{2}}{4}\sin(2\theta),\ \ \ \theta(0)=\alpha,\ \dot{\theta}(0)= \beta,\ t\geq 0,\ \varepsilon>0.$$
The null equilibrium solution is Lyapunov stable, explaining the stability of the Kapitza's inverted pendulum at the high frequency limit. See \cite{Evans_av} for a proof using weak convergence methods. Note that the first term of the effective equation is the average of the original equation while the second is due to high frequency effects.

A similar phenomenon as in the previous example occurs in our system where the high frequency effect as the Hooke's constant goes to infinity is the responsible for an effective potential which is the cause of the spontaneous spin reversion in Theorem \ref{main}. This will be proved in section \ref{Section_Proof}.

Looking forward for a self contained exposition, after the necessary preliminaries on weak convergence methods in subsection \ref{Weak_preliminaries}, in subsection \ref{Proof_Strong} we will calculate the effective Lagrangian of the system at the high frequency limit and prove the corresponding strong constraining result for our model.


The Lagrangian of the system will be calculated in section \ref{Section_Lagrangian}. As an immediate consequence of the Euler-Lagrange equations of this Lagrangian, we conclude that this holonomic model of the rattleback also reproduces the \textit{tapping effect}: even if the body is at rest with respect to some generic $\alpha$ angular coordinate, a non null initial velocity in the $\gamma$ angular coordinate implies a force in the $\alpha$ coordinate and the body starts to spin along the $z$ axis as well. Concretely,

\begin{rmk}[Tapping effect]\label{main3}
For every Hooke's constant $k>0$, the motion starting at the point $(0,\,\alpha_0)$ with angular velocity $(\dot{\gamma}_0,\,0)$ has acceleration at the starting point
$$\ddot{\gamma}_k(0)\,=\,0,\qquad \ddot{\alpha}_k(0)\,=\,i_z^{-1}\,(i_y-i_x)\,\cos(\alpha_0)\,\sin(\alpha_0)\,\dot{\gamma}_0^2.$$
In particular, if the body is asymmetric around the $z$ axis, that is to say $i_x\neq i_y$, $\alpha_0$ is not an integer multiple of $\pi/2$ and $\dot{\gamma}_0$ is non null, then the body has a non null angular acceleration along the $z$ axis.
\end{rmk}

\section{Lagrangian of the system}\label{Section_Lagrangian}

Denote by $X$, $Y$ and $Z$ the rotations about the respective coordinate axes in $\R^3$ and parametrize a rotation $R$ in $SO(3)$ as
\begin{equation}\label{angular_coord}
R=X(\gamma)\,Y(\beta)\,Z(\alpha).
\end{equation}

A motion in $SO(3)$ has the angular velocity
$$
\begin{pmatrix}
0 & -\omega_z & \omega_y\\
\omega_z & 0 & -\omega_x\\
-\omega_y & \omega_x & 0
\end{pmatrix}
=R^{-1}\,\dot{R}$$
in the Lie algebra $so(3)$. In terms of the angular coordinates \eref{angular_coord} we have
\begin{eqnarray*}
&& R^{-1}\,\dot{R}\, =\, R^{-1}\left(\partial_\alpha R\,\dot{\alpha}+\partial_\beta R\,\dot{\beta}+\partial_\gamma R\,\dot{\gamma}\right) \\
&=& Z^{-1}\partial_\alpha Z\,\dot{\alpha}+ Z^{-1}\left(Y^{-1}\partial_\beta Y\right) Z\,\dot{\beta} + Z^{-1}\,Y^{-1}\left(X^{-1}\partial_\gamma X\right) Y\,Z\,\dot{\gamma} \\
&=& L_z\,\dot{\alpha}+ Ad_{Z^{-1}}\,(L_y)\,\dot{\beta} + Ad_{(YZ)^{-1}}\,(L_x)\,\dot{\gamma}
\end{eqnarray*}
where $L_x$, $L_y$ and $L_z$ denote the usual basis of the Lie algebra $so(3)$ given by
$$L_x=\begin{pmatrix}
0 & 0 & 0\\
0 & 0 & -1\\
0 & 1 & 0
\end{pmatrix},\quad
L_y=\begin{pmatrix}
0 & 0 & 1\\
0 & 0 & 0\\
-1 & 0 & 0
\end{pmatrix},\quad
L_z=\begin{pmatrix}
0 & -1 & 0\\
1 & 0 & 0\\
0 & 0 & 0
\end{pmatrix}$$
and $Ad$ denotes the adjoint representation of the group $SO(3)$ in its Lie algebra
$$Ad:SO(3)\rightarrow Aut\left(so(3)\right),\quad Ad_M(L)=M\,L\,M^{-1}.$$
In particular, we have the following relation between the angular velocity vector $\omega$ and the angular coordinates
\begin{equation}\label{angular_velocity_relation}
\omega=\begin{pmatrix}
\omega_x\\
\omega_y\\
\omega_z
\end{pmatrix}=
\begin{pmatrix}
\cos(\alpha)\cos(\beta) & \sin(\alpha) & 0\\
-\sin(\alpha)\cos(\beta) & \cos(\alpha) & 0\\
\sin(\beta) & 0 & 1
\end{pmatrix}\,
\begin{pmatrix}
\dot{\gamma}\\
\dot{\beta}\\
\dot{\alpha}
\end{pmatrix}.
\end{equation}

The kinetic energy of a rigid body with moment of inertia tensor $I$ reads
$$K(\omega)=\frac{1}{2}\,\langle\,I\,\omega\,,\,\omega\,\rangle$$
where $\omega$ denotes the angular velocity vector. Imposing the holonomic constraint $\beta=0$ as well as the quadratic potential \eref{Hooke_constant} on the angular coordinate $\gamma$ and assuming that the moment of inertia tensor is diagonal, by the identity \eref{angular_velocity_relation} the Langrangian of the mechanical system reads as follows
\begin{equation}\label{Lagrangian}
L_k(\gamma,\, \alpha,\, v_\gamma,\, v_\alpha)= \frac{1}{2}\,(i_x\cos^2(\alpha)+i_y\sin^2(\alpha))\,v_\gamma^2\,+\,\frac{i_z}{2}\,v_\alpha^2\, -\,\frac{k}{2}\,\gamma^2.
\end{equation}

The Euler-Lagrange equations of the mechanical system are
\begin{eqnarray}\label{Euler_Lagrange_eqn}
\ddot{\gamma}_k\,&=&\,-\gi^{-1}\,\gi'\,\dot{\alpha}_k\,\dot{\gamma}_k\,-\,k\,\gi^{-1}\,\gamma_k \\
\ddot{\alpha}_k\,&=&\,\frac{i_z^{-1}}{2}\,\gi'\,\dot{\gamma}_k^2
\end{eqnarray}
where the function $\gi$ and its derivative are evaluated on $\alpha_k$ and this function is defined by
$$\gi(\alpha)\,=\,i_x\cos^2(\alpha)+i_y\sin^2(\alpha).$$

An immediate consequence of these equations is the tapping effect described in Remark \ref{main3}.

\section{Proof of Theorems \ref{main} and \ref{main2}}\label{Section_Proof}

In the next section we will prove the following

\begin{prop}\label{Strong_const_prop}
The solution $(\gamma_k,\,\alpha_k)$ of the Euler-Lagrange equations of the Lagrangian \eref{Lagrangian} with initial conditions
$$(\gamma_k(0),\,\alpha_k(0))\,=\,(0,\alpha_0),\qquad (\dot{\gamma}_k(0),\,\dot{\alpha}_k(0))\,=\,(\dot{\gamma}_0,\,\dot{\alpha}_0)$$
uniformly converges over compact sets to $(0,\,\alpha_{eff})$ where $\alpha_{eff}$ is the solution of the Euler-Lagrange equation of the Lagrangian
\begin{equation}\label{Lagrangian_eff}
L_{eff}(\alpha,\, v_\alpha)= \frac{i_z}{2}\,v_\alpha^2\, -\, \theta\,\left(i_x\,\cos^2(\alpha)\,+\,i_y\sin^2(\alpha)\right)^{-1/2}
\end{equation}
with initial conditions
$$\alpha_{eff}(0)\,=\,\alpha_0,\qquad \dot{\alpha}_{eff}(0)\,=\,\dot{\alpha}_0$$
where the constant $\theta$ is
$$\theta\,=\, \frac{\dot{\gamma}_0^2}{2}\,\left(i_x\,\cos^2(\alpha_0)\,+\,i_y\sin^2(\alpha_0)\right)^{3/2}.$$
Moreover, the convergence $\alpha_k\rightarrow\alpha_{eff}$ is $C^1$ over compact sets and
$$\dot{\gamma}_k^2\,\left(i_x\,\cos^2(\alpha_k)\,+\,i_y\,\sin^2(\alpha_k)\right)^{3/2}\,\rightharpoonup\, \theta.$$
weakly star over compact sets as $k\to +\infty$.
\end{prop}

In particular, the energy of the effective mechanical system with null initial condition $\alpha_0$ is
\begin{equation}\label{Energy_eff}
E_{eff}(\alpha,\, v_\alpha)= \frac{i_z}{2}\,v_\alpha^2\, +\, \frac{i_x}{2}\,\dot{\gamma}_0^2\,\left(\cos^2(\alpha)+(i_y/i_x)\,\sin^2(\alpha)\right)^{-1/2}.
\end{equation}


\begin{proof}[Proof of Theorem \ref{main}]
Because $\dot{\alpha}_0$ is non null, by the hypotheses the effective motion $\alpha_{eff}$ is subcritical hence is trapped going back and forth and bouncing on some interval $[-\alpha_*,\,\alpha_*]\subset [-\pi/2,\,\pi/2]$. Then, for every natural $n$ there is a time $T$ such that the effective motion $\alpha_{eff}$ goes back and forth $n$ times on the interval $[0,T]$. By the uniform convergence on $[0,T]$ as $k\to +\infty$, we have the result.
\end{proof}

\begin{proof}[Proof of Theorem \ref{main2}]
Because $\dot{\alpha}_0$ is non null, by the hypotheses the effective motion $\alpha_{eff}$ is supercritical hence it is not trapped and is strictly monotonic. Then, $\alpha_{eff}$ is strictly monotonic and for every $A>0$ there is a time $T>0$ such that $\vert\,\alpha_{eff}(T)\,\vert > A$. By the uniform convergence of the angular coordinates and the $C^1$ convergence of the $\alpha$ coordinate on $[0,T]$ as $k\to +\infty$, we have the result.
\end{proof}

\section{Strong constraining of the system}

This section is devoted to preliminaries on weak convergence and the proof of Proposition \ref{Strong_const_prop}.

\subsection{Preliminaries on weak convergence}\label{Weak_preliminaries}

This preliminary section contains the results needed for the proof of Proposition \ref{Strong_const_prop} in the next section. For the proofs and details, we refer the reader to the classical functional analysis reference \cite{Rudin} and to \cite{Adams} for Sobolev spaces and the references therein.

Consider a real Banach space $(E,\Vert\cdot\Vert_E)$ and its dual space $E^{*}$ consisting of bounded linear functionals of $E$. The operator norm $\Vert\cdot\Vert_{E^{*}}$ induces the \textit{strong topology} on $E^{*}$. The \textit{weak topoloqy} $\omega$ on $E^{*}$ is the coarsest topology such that every functional on $E^{**}$ is continuous. In particular,
$$f_n\mathop{\rightharpoonup}^{\omega} f\qquad if\qquad F(f_n)\rightarrow F(f),\ F\in E^{**}.$$
There is a canonical isometric embedding of the space $E$ on its bidual space
$$J:E\rightarrow E^{**},\qquad (x\mapsto \hat{x}),\qquad \hat{x}(h)= h(x).$$

The \textit{weak star topology} $\omega^{*}$ on $E^{*}$ is the coarsest topology such that every functional in $J(E)$ is continuous. In particular,
$$f_n\mathop{\rightharpoonup}^{\omega^{*}} f\qquad if \qquad f_n(h)\rightarrow f(h),\ h\in E.$$
If the space $E$ is reflexive, i.e if $J$ is an isomorphism, then the weak and weak star topologies on $E^{*}$ coincide. Moreover, the converse is also true.

As a direct application of the Banach-Steinhaus Theorem we have:

\begin{prop}\label{Bound}
Every weakly star convergent sequence in $E^{*}$ is bounded.
\end{prop}

\noindent As a direct application of the Banach-Alaoglu Theorem we have:

\begin{prop}\label{Alaoglu}
If $E$ is separable, then every bounded sequence in $E^{*}$ has a weakly star convergent subsequence.
\end{prop}

Now we specialize in the $L^{p}$ spaces over an interval $I= (-T,T)$ with $T>0$. Recall the isomorphisms
$$L^{\infty}(I)\cong L^{1}(I)^{*},\qquad L^{p}(I)\cong L^{q}(I)^{*},\qquad p,\,q> 1,\quad p^{-1}+q^{-1}=1,$$
where the evaluation is given by integration
\begin{equation}\label{evaluation}
x(h)\,=\, \int_{-T}^{T}\,dt\ x(t)\,h(t).
\end{equation}
Note that because the $L^{p}(I)$ spaces are reflexive for $1<p<\infty$, the weak and weak star topologies coincide in these cases. As an immediate corollary of Proposition \ref{Alaoglu} we have

\begin{cor}
Consider the $L^{p}(I)$ space such that $1<p\leq \infty$ with the weak star topology. Then, the closed unit ball is sequentially compact.
\end{cor}


Because the indicators over open sets in $(-T,T)$ generate a dense set in $L^q(I)$ with $1\leq q <\infty$, by definition of the weak star topology and \eref{evaluation} we have

\begin{cor}\label{Alaoglu_cor}
Consider the $L^{p}(I)$ space such that $1<p\leq \infty$ with the weak star topology. Then,
$$x_n\,\rightharpoonup\,x,\qquad\mbox{iff}\qquad \int_A\,x_n\,\rightarrow\,\int_A\,x_n$$
for every open set $A$ in $I$. Equivalently, dividing by the measure of the respective open set $A$, we conclude that $x_n\,\rightharpoonup\,x$
if and only if on every open set the average on the right hand side converges to the average on the left.
\end{cor}

As an example consider the sequence of functions $(\cos^2(x/n))_{n\in\N}$ in $L^\infty(I)$. This sequence does not converge in the strong topology. However, by the Riemann-Lebesgue Theorem, it weakly converges to the $1/2$ constant.
\begin{rmk}\label{Remark2}
In the previous example, note that although there are points whereat the sequence pointwise converge, these limits are in general different to the weak star limit, see the origin as an example whereat the sequence converges to one. We will see in the next section that a similar phenomenon occurs in Lemma \ref{Lema_weak_clave} where the weak star limit $\theta$ is half of the pointwise limit at the origin.
\end{rmk}


Every $L^{p}(I)$ space is clearly an $L^{\infty}(I)$--module with the strong topology. Actually, every $L^{p}(I)$ space with the weak star topology is also an $L^{\infty}(I)$--module as the next proposition shows.

\begin{prop}\label{module}
If $y_n\rightarrow y$ strongly in $L^{\infty}(I)$ and $x_n\rightharpoonup x$ weakly star in $L^{p}(I)$, then $y_n x_n\rightharpoonup y x$ weakly star in $L^{p}(I)$.
\end{prop}

A function $x$ in $L^{p}(I)$ has a \textit{weak derivative} $\dot{x}$ in $L^{p}(I)$ if
$$\dot{x}(h)= -x(\dot{h}),\qquad h\in C_c^{1}(I)$$
where $C_c^{1}(I)$ is the space of differentiable real valued functions with compact support on $I$. If a weak derivative exists, then it is unique.

Define the \textit{Sobolev space} $W^{1,p}(I)$ as the linear subspace of $L^{p}(I)$ whose elements have weak derivative in $L^{p}(I)$ and Sobolev norm
$$\Vert x\Vert_{W^{1,p}}= \Vert x \Vert_p +\Vert \dot{x}\Vert_p.$$
In the case where $p=2$, the Sobolev space is denoted by $H^{1}(I)$ and the norm is induced by the inner product
$$\langle\, x,\, y\,\rangle_{H^{1}}= \langle\, x,\, y\,\rangle_2 + \langle\, \dot{x},\, \dot{y}\,\rangle_2.$$

\begin{prop}\label{Sobolev_Banach}
The Sobolev space $W^{1,p}(I)$ is identified with a closed linear subspace of $L^{p}(I)\times L^{p}(I)$ with the strong topology under the map $x\mapsto (x,\dot{x})$. In particular, it is a Banach space.
\end{prop}

Identifying $W^{1,p}(I)$ with this closed linear subspace, the weak star topology on this space is defined as the subspace topology induced by the weak star topology of the product. In particular, because the weak star topology of the product is the product topology of the weak star topologies of the respective factors, we have that  $x_n\rightharpoonup x$ weakly star on $W^{1,p}(I)$ if and only if $x_n\rightharpoonup x$ and $\dot{x}_n\rightharpoonup \dot{x}$ weakly star on $L^{p}(I)$.

\begin{prop}\label{Alaoglu_W}
If $1<p\leq \infty$, then every bounded sequence in $W^{1,p}(I)$ has a weakly star convergent subsequence.
\end{prop}

\begin{prop}\label{magic}
Suppose that $1<p\leq \infty$ and consider a sequence $(x_n)$ in $W^{1,p}(I)$ such that $x_n\rightharpoonup {\bf 0}$ weakly star in $L^{p}(I)$ and the sequence of weak derivatives is uniformly bounded in $L^{p}(I)$. Then, $\dot{x}_n\rightharpoonup {\bf 0}$ weakly star in $L^{p}(I)$.
\end{prop}

Consider the Cauchy problem given by the ordinary differential equation and initial condition
\begin{equation}\label{Cauchy_problem_prel}
\dot{x}= F(t,x),\qquad x(0)= x_0
\end{equation}
such that $F$ is continuous and locally Lipschitz in the second variable. A \textit{weak solution} of the problem above is a solution $x$ in $W^{1,p}(I)$ of the integral equation
\begin{equation}\label{Cauchy_problem_prel_int}
x(t)= x_0+\int_0^{t}\, ds\ F(s,x(s)).
\end{equation}
Every differentiable solution in $C^{1}(\bar{I})$ of the problem \eref{Cauchy_problem_prel} is called a \textit{strong solution}. As an immediate corollary of the Sobolev embedding
$$W^{1,p}(I)\subset C(\bar{I})$$
and by Picard's Theorem assuring the local existence and uniqueness of a unique $C^1$ solution of \eref{Cauchy_problem_prel_int} among the continuous functions we have
\begin{prop}\label{weak_strong}
Every weak solution of \eref{Cauchy_problem_prel} is strong. In particular, it is unique.
\end{prop}

\subsection{Proof of Proposition \ref{Strong_const_prop}}\label{Proof_Strong}

\begin{lema}\label{Bounds_Lemma}
For every $k>0$, the solution $(\gamma_k,\,\alpha_k)$ of the Euler-Lagrange equations of the Lagrangian \eref{Lagrangian} with initial conditions
$$(\gamma_k(0),\,\alpha_k(0))\,=\,(0,\alpha_0),\qquad (\dot{\gamma}_k(0),\,\dot{\alpha}_k(0))\,=\,(\dot{\gamma}_0,\,\dot{\alpha}_0)$$
is defined over the whole real line and there is a constant $E_0$ independent of $k$ such that
$$\vert\dot{\gamma}_k(t)\vert,\ \vert\dot{\alpha}_k(t)\vert\ \leq\, (2\,i^{-1}\,E_0)^{1/2},\qquad\vert\gamma_k(t)\vert\,\leq\, (2\,k^{-1}\,E_0)^{1/2}$$
for every real $t$ where $i$ denotes the minimum of $i_x$, $i_y$ and $i_z$. In particular $(\gamma_k)_{k>0}$ uniformly converges to zero over the real line as $k$ goes to $+\infty$.
\end{lema}
\begin{proof}
Consider $k>0$ and suppose that the maximal interval of definition of the solution of the respective Euler-Lagrange equations is $(\omega_-,\,\omega_+)$. On this interval, the solution is smooth and the energy of the system is conserved
$$E_k\,=\, \frac{1}{2}\,\gi\,\dot{\gamma}_k^2\,+\,\frac{i_z}{2}\,\dot{\alpha}_k^2\, +\,\frac{k}{2}\,\gamma_k^2$$
where $\gi$ is evaluated at $\alpha_k$. In particular, this energy is independent of $k$ and equals the initial energy $E_0$. Because every term of the energy function is positive, each one is bounded from above by $E_0$ and the bounds follow immediately.

In particular, the angular coordinates are bounded by
$$\vert\gamma_k(t)\vert\,=\,\left\vert\int_0^t ds\,\dot{\gamma}_k(s)\right\vert\,\leq\,\int_0^t ds\,\vert\dot{\gamma}_k(s)\vert\,\leq\,(2\,i^{-1}\,E_0)^{1/2}\,t$$
and an analogous argument holds for $\alpha_k(t)$.

Suppose that $\omega_+$ is finite. Then, the graph of the solution restricted to the interval $[0,\,\omega_+)$ is contained on a compact set which is absurd hence $\omega_+$ is infinite and an analogous argument holds for $\omega_-$. This concludes the proof.
\end{proof}

\begin{lema}
Let $T>0$. There is a sequence $(k_n)$ and a $C^1$ function $\alpha$ defined on the interval $[-T,T]$ such that the sequence $(\alpha_n)\,=\,(\alpha_{k_n}|_{[-T,T]})$ converges in the $C^1$ topology to $\alpha$ as $n$ goes to $+\infty$. In particular, $\alpha(0)\,=\,\alpha_0$ and $\dot{\alpha}(0)\,=\,\dot{\alpha}_0$.
\end{lema}
\begin{proof}
Consider the family of functions $\alpha_{k}|_{[-T,T]}$ and note that this is a family of Lipschitz functions uniformly bounded by $(2\,i^{-1}\,E_0)^{1/2}\,T$ with Lipschitz constant $(2\,i^{-1}\,E_0)^{1/2}$ independent from $k$. By the Arzel\`a--Ascoli Theorem, there is a sequence $\left(\alpha_{k_n}|_{[-T,T]}\right)$ of this family uniformly converging to some continuous function $\alpha$.

Denote by $\alpha_n$ the function $\alpha_{k_n}|_{[-T,T]}$ and consider the family of its derivatives. Again this is a family of Lipschitz functions uniformly bounded by $(2\,i^{-1}\,E_0)^{1/2}$ with Lipschitz constant independent from $k$ for
$$\vert\ddot{\alpha}_n(t)\vert\,\leq\,\vert i_z^{-1}\,(i_y-i_x)\,\cos(\alpha_n(t))\sin(\alpha_n(t))\,\dot{\gamma}_n^2\vert
\,\leq\,2\,i_z^{-1}\,(i_y-i_x)\,i^{-1}\,E_0$$
where we have used the Euler-Lagrange equation for the $\alpha$ coordinate. By the Arzel\`a--Ascoli Theorem, taking a subsequence if necessary, we may suppose without loss of generality that the sequence $(\dot{\alpha}_n)$ uniformly converges to some continuous function $\beta$.

We conclude that $\alpha$ is $C^1$ and $\beta$ coincides with $\dot{\alpha}$ for
$$\alpha_n(t)\,=\,\int_0^t ds\,\dot{\alpha}_n(s),\qquad t\in [-T,T]$$
and taking the limit on both sides and recalling that these convergences are uniform we have
$$\alpha(t)\,=\,\int_0^t ds\,\beta(s),\qquad t\in [-T,T]$$
and this concludes the proof.
\end{proof}

The rest of the proof consists in the calculation of the limit function $\alpha$.

\begin{lema}\label{Lema_weak_clave}
There is a constant $\theta$ such that taking a further subsequence if necessary we have the weakly star convergence
$$\dot{\gamma}_n^2\,\rightharpoonup\,\theta\,\gi(\alpha)^{-3/2},\qquad n\to+\infty.$$
Moreover, the constant $\theta$ equals
$$\theta\,=\,\frac{1}{2}\,\dot{\gamma}_0^2\,\gi(\alpha_0)^{3/2}.$$
\end{lema}
\begin{proof}
Define the transversal energy by
$$E_n^\perp\,=\, K_n^\perp\,+\,U_n^\perp,\qquad K_n^\perp\,=\,\frac{\gi(\alpha_n)}{2}\,\dot{\gamma}_n^2,\qquad U_n^\perp\,=\,\frac{k_n}{2}\,\gamma_n^2.$$
Because the sequences $(U_n^\perp)$ and $(\dot{\gamma}_n^2)$ are uniformly bounded, taking a further subsequence if necessary, by Corollary \ref{Alaoglu_cor} we may suppose without loss of generality that they weakly star converge to some $\sigma$ and $\pi$ in $L^{\infty}[-T, T]$ respectively
$$U_n^\perp\,\rightharpoonup\,\sigma,\quad\qquad \dot{\gamma}_n^2\,\rightharpoonup\,\pi.$$

The Euler-Lagrange equation for the $\gamma$ coordinate has the form
$$\ddot{\gamma}_n\,=\,b_n\,-\,k_n\,\gi(\alpha_n)^{-1}\,\gamma_n$$
where $b_n$ is a uniformly bounded sequence. Then, multiplying on both sides by $\gamma_n/2$ and taking the limit we have
$$\ddot{\gamma}_n\,\gamma_n/2\,=\,b_n\,\gamma_n/2\,-\,\gi(\alpha_n)^{-1}\,U_n^\perp\,\rightharpoonup\,-\gi(\alpha)^{-1}\,\sigma.$$

In particular, the sequence $\left(\dot{\gamma}_n\,\gamma_n/2\right)$ uniformly converges to zero and the respective derivatives are uniformly bounded hence by Proposition \ref{magic}, we have
$$\frac{d}{dt}\left(\dot{\gamma}_n\,\gamma_n/2\right)\,=\,\ddot{\gamma_n}\,\gamma_n/2\,+\,\dot{\gamma}_n^2/2\,\rightharpoonup\,{\bf 0}$$
and by the uniqueness of the limit we conclude the equipartition of the energy at the limit
$$\sigma\,=\,\frac{\gi(\alpha)}{2}\,\pi.$$

We claim that taking a subsequence if necessary, we may suppose that the previous sequence actually converges in the Sobolev space $W^{1,\infty}[-T,T]$. In effect,
$$E_n^\perp\,\rightharpoonup\,\frac{\gi(\alpha)}{2}\,\pi\,+\,\sigma\,=\,\gi(\alpha)\,\pi,$$
$$\dot{E}_n^\perp\,=\,\frac{\gi'(\alpha_n)}{2}\,\dot{\alpha}_n\,\dot{\gamma}_n^2\,+\,\gi(\alpha_n)\,\dot{\gamma}_n\,\ddot{\gamma}_n\, +\,
k_n\,\gamma_n\,\dot{\gamma}_n$$
$$=\,\dot{\gamma}_n\,\left(\gi'(\alpha_n)\,\dot{\alpha}_n\,\dot{\gamma}_n\,+\,\gi(\alpha_n)\,\ddot{\gamma}_n\, +\,
k_n\,\gamma_n\right)\,-\,\frac{\gi'(\alpha_n)}{2}\,\dot{\alpha}_n\,\dot{\gamma}_n^2$$
$$\rightharpoonup\,\frac{\gi'(\alpha)}{2}\,\dot{\alpha}\,\pi$$
where we have used the Euler-Lagrange equation for the $\alpha$ coordinate. Therefore, by Propositions \ref{Bound} and \ref{Sobolev_Banach}, the sequence $(E_n^\perp)$ is bounded in $W^{1,\infty}[-T,T]$ and by Proposition \ref{Alaoglu_W}, there is a subsequence converging with respect to the corresponding Sobolev norm to some element $E^\perp$ in this space verifying
$$E^\perp\,=\,\gi(\alpha)\,\pi,\qquad\quad \dot{E}^\perp\,=\, \frac{\gi'(\alpha)}{2}\,\dot{\alpha}\,\pi$$
where $\dot{E}^\perp$ is the weak derivative of $E^\perp$. In particular, $\pi$ is a weak solution of the following ordinary differential equation in the Sobolev space $W^{1,\infty}[-T,T]$
$$\dot{\pi}\,=\,-\frac{3}{2}\,\frac{d\,\log(\gi(\alpha))}{dt}\,\pi.$$

Because every weak solution of a Lipschitz ordinary differential equation in the mentioned Sobolev space is actually strong by Proposition \ref{weak_strong}, we conclude that $\pi$ is $C^1$ and there is a constant $\theta$ such that $\pi\,=\,\theta\,\gi(\alpha)^{-3/2}$.

Now we calculate the constant $\theta$. Recall that the initial energy $E_0$ is the constant independent of $n$
$$E_0\,=\,\frac{i_z}{2}\,\dot{\alpha}_0^2\,+\,\frac{\gi(\alpha_0)}{2}\,\dot{\gamma}_0^2.$$
For every $n$ the energy coincides with the initial energy hence
$$E_0\,=\,\frac{i_z}{2}\,\dot{\alpha}_n^2\,+\,E_n^\perp\,\rightharpoonup\,\frac{i_z}{2}\,\dot{\alpha}^2\,+\,\theta\,\gi(\alpha)^{-1/2}$$
and by the uniqueness of the limit both sides are equal. Evaluating the right hand side at time zero, we have the expression for $\theta$. This concludes the proof.
\end{proof}

\begin{proof}[Proof of Proposition \ref{Strong_const_prop}]
Consider the sequence in Lemma \ref{Lema_weak_clave}. The Euler-Lagrange equation for the $\alpha$ coordinate in integral form is
$$\dot{\alpha}_n(t)\,=\,\dot{\alpha}_0\,+\,\frac{i_z^{-1}}{2}\,\int_0^t ds\ \gi'(\alpha(s))\,\dot{\gamma}_n(s)^2$$
and because of the weakly star convergence of $(\dot{\gamma}_n^2)$, taking the limit on both sides we have
$$\dot{\alpha}(t)\,=\,\dot{\alpha}_0\,+\,\theta\,\frac{i_z^{-1}}{2}\,\int_0^t ds\ \gi'(\alpha(s))\,\gi(\alpha(s))^{-3/2}.$$
We conclude that actually $\alpha$ is at least $C^2$ and verifies the differential equation
\begin{equation}\label{Cauchy_problem}
\ddot{\alpha}\,=\,\theta\,\frac{i_z^{-1}}{2}\,\gi'(\alpha)\,\gi(\alpha)^{-3/2},\qquad \alpha(0)=\alpha_0,\quad\dot{\alpha}(0)=\dot{\alpha}_0
\end{equation}
which is the Euler-Lagrange equation of the effective Lagrangian \eref{Lagrangian_eff}.

We have proved that for every $T>0$ and every sequence in the family $(\alpha_k)_{k>0}$ there is a subsequence converging in $C^1[-T,T]$ to the unique solution $\alpha$ of the Cauchy problem \eref{Cauchy_problem}. We conclude that the whole family converges to this unique solution. Indeed, if there is a sequence converging to some function other than $\alpha$, there is a subsequence converging to $\alpha$ which absurd for the limits of the sequence and its subsequence must coincide.

Because the constant $\theta$ in Lemma \ref{Lema_weak_clave} is unique, an analogous argument as the previous one shows the weakly star convergence of the whole family to this constant. This concludes the proof.
\end{proof}

\ack

The author is a research fellow at \textit{Consejo Nacional de Ciencia y Tecnolog\'ia}.

\References

\bibitem[Ad]{Adams}
Adams R A 1975, \emph{Sobolev Spaces}, Academic Press, New York.

\bibitem[Ar]{Arnold}
Arnold V I 1978, \emph{Mathematical methods of classical mechanics}, Springer-Verlag, Berlin, Heidelberg, New York.

\bibitem[ACC]{Tennis_racket}
Ashbaugh M S, Chicone C C, Cushman R H 1991, \emph{The twisting tennis racket}, Journal of Dynamics and Differential Equations {\bf 3} 67--85.

\bibitem[AKN]{Arnold_Kozlov}
Arnold V I, Kozlov V V, Neishtadt A I 2006, \emph{Mathematical Aspects of Classical and Celestial Mechanics}, Springer Link, Third Edition.

\bibitem[Bo]{2}
Bondi H 1986, \emph{The rigid body dynamics of unidirectional spin}, Proc. R. Soc. Lond. A {\bf 405} 265--274.

\bibitem[Bo]{Bornemann}
Bornemann F A 1998, \emph{Homogenization in Time for Singularly Perturbed Mechanical Systems}, Lecture Notes in Mathematics {\bf 1687}, Springer.

\bibitem[Bu]{Burgos}
Burgos J M, \emph{Strong degenerate constraining in Lagrangian dynamics}, arXiv: 2104.06549. 

\bibitem[BKK]{3}
Borisov A V, Kazakov A O, Kuznetsov S P 2014, \emph{Nonlinear dynamics of the rattleback: a nonholonomic model}, Physics-Uspekhi {\bf 57} (5) 453--460.

\bibitem[BM]{4}
Borisov A V, Mamaev I S 2003, \emph{Strange attractors in rattleback dynamics}, Physics-Uspekhi {\bf 46} (4) 393--403.

\bibitem[BS]{BS}
Bornemann F A, Sch\"utte C 1997, \emph{Homogenization of Hamiltonian systems with a strong constraining potential}, Phys. D, {\bf 102}, 57--77.

\bibitem[EZ]{Evans_av}
Evans L C, Zhang T 2016, \emph{Weak convergence and averaging for ODE}, Nonlinear Analysis: Theory, Methods and Applications {\bf 138}, 83--92.

\bibitem[Ga]{Gallavotti}
Gallavotti G 1983, \emph{The Elements of Mechanics}, Springer-Verlag, Berlin, Heidelberg, New York.

\bibitem[GH]{7}
Garcia A, Hubbard M 1988, \emph{Spin reversal of the rattleback: theory and experiment}, Proc. R. Soc. Lond. A {\bf 418} 165--197.

\bibitem[KN]{8}
Kondo Y, Nakanishi H 2017, \emph{Rattleback dynamics and its reversal time of rotation}, Phys. Rev. E {\bf 95} 062207.


\bibitem[LL]{LL}
Landau L D, Lifshitz E M 1976, \emph{Mechanics}, Course of Theoretical Physics, Volume 1, Third Edition, Butterworth--Heinemann.

\bibitem[MT]{9}
Moffatt H K, Tokieda T 2008, \emph{Celt reversals: a prototype of chiral dynamics}, Proc. Royal Soc. Edinburgh {\bf 138A} 361--368.

\bibitem[Ru]{Rudin}
Rudin W 1976, \emph{Principles of Mathematical Analysis}, Madison, WI.

\bibitem[RU]{RubinUngar}
Rubin H, Ungar P 1957, \emph{Motion under a strong constraining force}, Comm. Pure. Applied Math., {\bf 10}, 65--87.

\bibitem[RWP]{12}
Rauch-Wojciechowski S, Przybylska M 2017, \emph{Understanding reversals of a rattleback}, Regul. Chaotic Dyn. {\bf 22} (4) 368--385.

\bibitem[Ta]{Takens}
Takens F 1980, \emph{Motion under the influence of a strong constraining potential}, Global Theory of Dynamical Systems, Z. Nitecki and C. Robinson eds., Springer-Verlag, Berlin, Heidelberg, New York, 425--445.

\bibitem[TG]{Integrable_rattleback}
Tudoran R M, Gîrban A 2020, \emph{On the rattleback dynamics}, Journal of Mathematical Analysis and Applications {\bf 488} (1)  124066.

\bibitem[Wa]{16}
Walker G T 1895, \emph{On a curious dynamical property of celts}, Proc. Camb. Philos. Soc. {\bf 8} 305--306.

\bibitem[YTM]{15}
Yoshida Z, Tokieda T, Morrison P J 2017, \emph{Rattleback: A model of how geometric singularity induces dynamic chirality}, Phys. Lett. A {\bf 381} (34) 2772--2777.

\endrefs

\end{document}